\documentclass[]{amsart}

\usepackage{graphicx}

\newtheorem{thm}{Theorem}

\newtheorem{corollary}{Corollary}
\newtheorem{lem}{Lemma}
\newtheorem{lemma}{Lemma}

\newtheorem{fact}{Fact}

\theoremstyle{definition}

\newtheorem{remark}{Remark}

\newtheorem{example}{Example}

\newtheorem{definition}{Definition}

\begin{document}

\title[Second moment method for a family of boolean CSP]{Second moment method for a family of boolean CSP}

\author[Y. Boufkhad]{Yacine Boufkhad}
\address[Y. Boufkhad]{  LIAFA, Universit{\'e} Paris Diderot, CNRS UMR 7089 
  \newline Paris, France}  
\email{ Yacine.Boufkhad@liafa.jussieu.fr}  

\author[O. Dubois]{Olivier Dubois}
\address[O. Dubois]{LIP6, CNRS UMR7606, Universit{\'e} Pierre et Marie Curie	
  \newline Paris, France} 
\email{ Olivier.Dubois@lip6.fr}  




\setlength{\parindent}{0in}

\begin{abstract}
The estimation of phase transitions in random boolean Constraint Satisfaction
Problems (CSP) is based on two fundamental tools: the first and second
moment methods. While the first moment method on the number of solutions
permits to compute upper bounds on any boolean CSP, the second moment
method used for computing lower bounds proves to be more tricky and
in most cases gives only the trivial lower bound 0. In this paper,
we define a subclass of boolean CSP covering the monotone versions
of many known NP-Complete boolean CSPs. We give a method for computing
non trivial lower bounds for any member of this subclass. This is
achieved thanks to an application of the second moment method to some
selected solutions called characteristic solutions that depend on
the boolean CSP considered. We apply this method with a finer analysis to
establish that the threshold $r_{k}$ (ratio : \#constrains/\#variables) of monotone 1-in-k-SAT is $\log k/k\leq r_{k}\leq\log^{2}k/k$.
\end{abstract}
\maketitle

\section*{Introduction}

The empirical evidence has shown that random instances of boolean
Constraint Satisfaction Problems $CSP$ exhibit a phase transition
i.e. a sudden change from SAT to UNSAT when the number of constraints
increases: that is there exists a critical value $r^{*}$ of the ratio
$r$ number of constraints to number of variables such that random
instances are w.h.p. satisfiable if $r<r^{*}$ and w.h.p. unsatisfiable
if $r>r^{*}.$ $r^{*}$ is called the threshold value of the transition.
The sharpness of the threshold has been addressed in a series of works
\cite{Sharpthresholdsforgraphpropertiesandthek-SATproblem::1999::Friedgut,DBLP:conf/stoc/Molloy02a,CombinatorialsharpnesscriterionandphasetransitionclassificationforrandomCSPs.::2004::CreignouDaude,DBLP:journals/dm/CreignouD09}.

Computing the threshold associated to a $CSP$ is at present out of
reach apart from some exceptions \cite{MickGetsSometheOddsAreonHisSide::1992::ChvatalReed,AThresholdforUnsatisfiability.::1996::Goerdt,DelaVega,The3-XORSATthreshold::2002::DuboisMandler}
for polynomial subclasses. Since this cannot be carried out, upper
and lower bounds of $r^{*}$ are computed. These bounds are almost
all obtained by different applications of the probabilistic method
tools: the first and second moment methods and for some of them through
anlaysis of algorithms\cite{Typicalrandom3-SATformulaeandthesatisfiabilitythreshold::2003::DuboisBoufkhadMandler,Typicalrandom3-SATformulaeandthesatisfiabilitythreshold::2000::DuboisBoufkhadMandler,DBLP:journals/tcs/DiazKMP09,Theprobabilisticanalysisofagreedysatisfiabilityalgorithm.::2006::KaporisKirousisLalas,DBLP:conf/sat/BoufkhadH10}
for $3$-SAT. 

While the first moment method on the number of solutions permits to
obtain an upper bound of the location of the threshold for any boolean
$CSP$, the second moment method on the number of solutions fails at any ratio to estimate the probability of satisfiability. In \cite{TheThresholdforRandomk-SATis2mboxkln2-Ok::2004::AchlioptasPeres}, an original method is presented to overcome this problem in the case of $k$-SAT for which the direct calculus also fails. 

We define a subclass of $CSP$ and a method that allow
to bound the phase transition from both sides for this subclass. The
latter is characterized by constraints having the property of being
closed under permutations. It includes the monotone versions of many
well known problems like : 1-in-kSAT, NAE-k-SAT... and then it includes
many NP-Complete boolean $CSP$s. 

Roughly speaking, we show how the second method can be made ``to
work'' for every problem in this class. More precisely, we parameterize the valuations by their number of variables having the value $1$ and we show that  there exist precise values for this parameter depending on every problem for which
the second moment method gives a non trivial lower bound, the corresponding solutions are called \emph{characteristic solutions}. 
We prove that the bound given by this method is at least some well defined value for any problem in the sublclass. However, the generality of this value is obtained at the cost of some weakness.  
Better bounds can be computed using the same scheme through a finer analysis on a case by case basis. To illustrate this, we do the full analysis  for positive $1$-in-$k$-SAT
to derive the asymptotically optimal lower bound with respect to our method that is $\log k/k$. To
show that this lower bound is tight, we establish using the first
moment method an upper bound of $\log^{2}k/k$ for $k\geq7$.


\section{\label{sec:Definitions-and-outline}Basic definitions and main results}

Given a set $X$ of $n$ boolean variables, a valuation $\sigma$
is a mapping $X\rightarrow\{0,1\}$ that assigns to any variable $x\in X$
the value $0$ or $1$. $k$ being an integer, a relation $R$ of
arity $k$ is a subset of $\{0,1\}^{k}$. A relation is said to be
trivial if $\left(0,...,0\right)$ or $\left(1,...,1\right)$ is an
element of $R$. We consider throughout  the paper only non trivial
relations. A constraint defined from a relation $R$ of arity $k$,
  is a tuple of $k$ boolean variables, denoted $R\left(x_{1},...,x_{k}\right)$,
which is said to be satisfied under some valuation $\sigma$ iff
$\left(\sigma\left(x_{1}\right),...,\sigma\left(x_{k}\right)\right)\in R$,
otherwise it is said unsatisfied. Given a set of relations $S$,
an instance of boolean CSP with respect to $S$, denoted by $CSP\left(S\right)$,
is a conjunction set of constraints $R\left(x_{1},...,x_{k}\right)$
where $R\in S$. An instance $CSP\left(S\right)$ is satisfied iff
every constraint is satisfied. 

In this paper, we consider a subclass of $CSP\left(S\right)$s defined as follows. A relation $R$ is said to
be invariant by permutation, iff any permutation of the coordinates
of a tuple $t\in R$ is also in $R.$ Such a relation is denoted by
$R_{inv}$. The invariance property implies that for every tuple $t\in R_{inv}$,
all tuples having the same number of coordinates equal to 1 or (0s)
as $t$ must be also in $R_{inv}.$ This defines an equivalence relation,
two elements $t$ and $t'$ of $R_{inv}$ belonging to the same equivalence
class iff they have the same number of coordinates equal to $1$s
(or 0s). Thus the equivalence classes partition $R_{inv}$ into subsets
each associated to an integer $i$ equal to the number of $1$s of
the elements of the class. 

In this paper, we will only consider boolean $CSP$s with respect
to a single non trivial invariant under permutation relation denoted
$CSP(\{R_{inv}\})$ . In order to designate more explicitly a $CSP(R_{inv})$
 we will denote it in an equivalent manner by $CSP\left(I_{k}\right)$,
where $I_{k}$ is the subset of integers in $\{1,...,k-1\}$ associated
to all equivalence classes. Thus an instance
$CSP(I_{k})$ is satisfiable with respect to $I_{k}$ iff there exits
a valuation $\sigma$ such that the number of $1$s in every constraint
of $CSP(I_{k})$\footnote{Alternatively, this class can be seen as hypergraph bi-coloring problem where the number of vertices allowed to have a certain color in some edge are taken only in $I_k$. } is an element of $I_{k}$. 
\begin{example}
$k=4$ and $R=\{1000,0100,0010,0001,0111,1011,1101,1110\}$. $R$
is invariant by permutation. The set of integers associated to $R$ is
$I_{4}=\{1,3\}$. A constraint $(x_{i_{1}},x_{i_{2}},x_{i_{3}},x_{i_{4}})$
of an instance $CSP(R)$ is satisfied iff exactly one or three of
the four variables of the constraint has the value 1. 

The $CSPs$ of the class defined above are NP-complete for any relation
of arity greater or equal to 3 according to the Schaefer classification.
They include two well known problems of this classification that
are positive $1$-in-$k$-SAT ($I_{k}=\{1\}$ according to the above
definition) and positive not-all-equal-$k$-SAT ($I_{k}=\{1,...,k-1\}$). 
\end{example}
The random version of a $CSP(I_{k})$ is as follows. Given a relation
$R_{inv}$, $I_{k}$ the set of integers associated with $R_{inv}$,
a random $CSP(I_{k})$ instance with $m$ constraints over $n$ boolean
variables is formed by drawing uniformly, independently and with replacement
$m$ tuples of $k$ variables over the set of $n$ variables. Such
a random $CSP(I_{k})$ instance is denoted by $I_{k}(m,n).$ This
defines a probability space denoted by $\Omega(I_{k},m,n)$ in which
instances $I_{k}(m,n)$ are equiprobable. 
\begin{definition}
A $p$-valuation for some natural integer $0\leq p\leq n$ is a valuation
such that $|\{x_{i}|\sigma\left(x_{i}\right)=1\}|=p$. 
\end{definition}
Let $\delta\in[0,1]$, for the sake of simplicity we will denote whenever
it is non ambiguous a $\lfloor\delta n\rfloor$-valuation by $\delta$-valuation.
A $\delta$-valuation that is a solution of an instance $I_{k}(m,n)$
is said to be a $\delta$-solution. Let $X_{\delta}$ be the random
variable associating to each $I_{k}(m,n)$ the number of its $\delta$-solutions. 
\begin{thm}
\label{lem:secondmoment}For any $CSP\left(I_{k}\right)$, there exist
a $r^*_{I_{k}}>0$ and $0<\delta<1$ such that for all $r<r^*_{I_{k}}$,
\textup{$lim_{n\rightarrow\infty}\frac{{\bf E}[X_{\delta}]^{2}}{{\bf E}[X_{\delta}^{2}]}>0$.}
\end{thm}
Roughly speaking, the preceding Theorem states that for any problem
in $CSP\left(I_{k}\right)$, there exists a $\delta$ that makes the
second moment to succeed in computing a lower bound. Combined with
the inequality of Cauchy-Schwartz: 

\[
{\bf P}\left(X_{\delta}>0\right)\geq\frac{{\bf E}[X_{\delta}]^{2}}{{\bf E}[X_{\delta}^{2}]}\]

\noindent and the sharpness of the threshold of the problems in this class \cite{EveryMonotoneGraphPropertyhasaSharpThreshold::1996::FriedgutKalai}
\cite{Generalizedsatisfiabilityproblems:minimalelementsandphasetransitions.::2003::CreignouDaude},
we  have the following consequence of Theorem \ref{lem:secondmoment}.

\begin{corollary}
\label{thm:borneinf}For any $CSP\left(I_{k}\right)$, there exists
a real $r^*_{I_{k}}>0$ such that for $r<r^*_{I_{k}}$, 

${\displaystyle lim_{n\rightarrow\infty}}Pr\left(I_{k}\left(n,rn\right)\mbox{ is satisfiable }\right)=1$.
\end{corollary}
The Theorem \ref{lem:secondmoment} states that the second moment succeeds for some values of $\delta$ using $X_\delta$ as a random variable. However, the value of the bound $r^*_{I_k}$ mentioned in the theorem (and given in Section \ref{preuve13}) is not the optimal bound that can be derived by the method. This is the price of its generality. 
The full analysis establishing the optimal bound can somehow be done on a case by case basis. Note that the best bound that can be obtained can not exceed the smallest ratio that make $\mathbf{E}[X_{\delta}] \rightarrow 0$ when $n\rightarrow \infty$. Indeed, by Markov inequality the probability of $X_\delta>0$ in that case is $0$. The value of the best bound that can be expected for  positive 1-in-k-SAT is asymptotically $\log k/k$. This is precisely what we obtain through the full analysis of this particular problem. We obtain : 
\begin{thm}
\label{thm:1ksat}Let $1_{k}=\{1\}$. $lim_{n\rightarrow\infty}{\bf Pr}\left(1_{k}\left(n,rn\right)\mbox{ is satisfiable}\right)=1$ If $r<\log k/k$ and $k\geq3$. 
\end{thm}
Specifically,  for $k=3$ a better lower bound at 0.546 has been computed in \cite{DBLP:journals/corr/abs-cs-0508037}
analyzing the success to find out a solution with a specific algorithm.
However our aim is to provide a tool yielding systematically a lower
bound for a large class of CSPs. We give a rough upper bound at $(\log k)^{2}/k$
(valid for $k\geq7$) which shows that our bounds are tight around
the threshold.

\section{\label{sec:General-case}Second moment and characteristic solutions}
We first define the \emph{characteristic solutions} before we give the second moment of their number. 
\begin{definition}
\emph{Characteristic valuations} for some $CSP\left(I_{k}\right)$
are $\delta$-valuations for which the probability of satisfying a
uniformly randomly drawn constraint is locally maximum with respect
to $\delta$. The solutions of an instance that are characteristic
valuations are said to be \emph{characteristic solutions} for this
instance. 
\end{definition}
Given some $\delta$-valuation, the probability $\pi_{i}\left(\delta\right)$
that a randomly selected $k$-tuple contains $i$ ones is $\pi_{i}\left(\delta\right)={k \choose i}\delta^{i}\left(1-\delta\right)^{k-i}$.
So the probability that a $\delta$-valuation satisfies with respect
to some set $I_{k}\subset\{1,2,...,k-1\}$ a randomly selected $k$-tuple
is obtained by summing up, the mutually exclusive cases for different
$i$'s in $I$. This probability is $g_{I_{k}}\mbox{\ensuremath{\left(\delta\right)}}=\sum_{i\in I_{k}}\pi_{i}\left(\delta\right)=\sum_{i\in I_k}{k \choose i}\delta^{i}\left(1-\delta\right)^{k-i}$.
Let $\Delta_{I_{k}}$ be the set of reals for which $g_{I_{k}}\left(\delta\right)$
is locally maximum. Clearly, for any $\delta\in\Delta_{I_{k}}$, $\delta$-valuations
are by definition the characteristic valuations of $CSP\left(I_{k}\right)$.

Since $I_{k}\subseteq\{1,...,k-1\}$ then $g_{I_{k}}\left(0\right)=g_{I_{k}}\left(1\right)=0$.
The function $g_{I_{k}}\left(\delta\right)$ being smooth, strictly
positive inside $]0,1[$, it maximizes inside the interval $]0,1[$
at at least one stationary point. Thus for any $I_k \subseteq\{1,...,k-1\}$, $\Delta_{I_{k}}\neq\emptyset$. The fact that every $\delta \in\Delta_{I_{k}}$ is a stationary point for $g_{I_{k}}\left(\delta\right)$ will be used later.

\subsection{\label{sub:Second-moment-method}First and second moment of number
of characteristic solutions}

The key idea used in the method presented in this paper is instead
of taking as a random variable the number of solutions, to consider as random variable the number of $\delta$-solutions. The first moment
of $X_{\delta}$ is:

\[
\mathbf{E}[X_{\delta}]={n \choose \delta n}g_{I_{k}}\mbox{\ensuremath{\left(\delta\right)}}^{rn}\sim\frac{1}{\sqrt{2\pi\delta(1-\delta)n}}\gamma_{I_k,r}\left(\delta\right)^{n}\]

Where:\[
\gamma_{I_{k},r}\left(\delta\right)=\frac{g_{I_{k}}\left(\delta\right)^{r}}{\delta^{\delta}\left(1-\delta\right)^{1-\delta}}\]

\begin{remark}
\label{best}
 It is easy to see that $\lim_{n \rightarrow \infty} \gamma_{I_{k},r}\left(\delta\right)=0$ for any  $r>\hat{r}_{I_k,\delta}=\frac{\delta \log \delta+(1-\delta)\log(1-\delta)}{\log g_{I_k}(\delta)}$. By Markov inequality, this means that the $\delta$-solutions does not exist for  $r>\hat{r}_{I_k,\delta}$ and then the lower bound that we can get through $\delta$-solutions is at most $\hat{r}_{I_k,\delta}$.
\end{remark}

For computing the second moment, we consider two
$\delta$-valuations $\sigma_{1}$ and $ $$\sigma_{2}$ having $p$
variables assigned $1$ in $ $$\sigma_{1}$ and $0$ in $\sigma_{2}$.
This defines all the other categories of variables. Indeed, the
number of variables assigned $0$ in $\sigma_{1}$ and $1$ in
$\sigma_{2}$ must be also $p$ in order that $\sigma_{2}$ is a
$\delta$-valuation. $\lfloor\delta n\rfloor-p$ is the number of
variables assigned $1$ in both solutions and $n-\lfloor\delta
n\rfloor-p$ are assigned $0$ in both solutions. First, we give the
probability $\phi_{i,j,\delta}\left(\frac{p}{\lfloor\delta
n\rfloor}\right)$ that a random $k$-tuple has $i$ $1$s under
$\sigma_{1}$ and $j$ $1$s under $\sigma_{2}$ such that $d$
variables of the $k$-tuple are equal to $1$ in both assignments.
$d$ must range from $d_{min}=max(0,i+j-k)$ to $d_{max}=min(i,j)$.

\begin{eqnarray}
\label{phi} \phi_{i,j,\delta}\left(\frac{p}{\lfloor\delta
n\rfloor}\right) & = & \sum_{d_{min}}^{d_{max}}{k \choose i}{i
\choose d}{k-i \choose j-d}\left(\frac{\lfloor\delta
n\rfloor-p}{n}\right)^{d}\left(\frac{p}{n}\right)^{i+j-2
d}\left(\frac{n-\lfloor\delta n\rfloor-p}{n}\right)^{k-i-j+d}
\end{eqnarray}
Summing up over couples $\left(i,j\right)$, we get
$G_{I_k,\delta}\left(\mu\right)$, the probability that a couple of
$\delta$-valuations having $\mu\delta n$ variables taking a
different value in $\sigma_{1}$ or $\sigma_{2}$ satisfies a random
constraint:

\begin{eqnarray*}
G_{I_{k},\delta}\left(\frac{p}{\lfloor\delta n\rfloor}\right) & =
& \sum_{i\in I_{k}}\sum_{j\in
I_{k}}\phi_{i,j,\delta}\left(\frac{p}{\lfloor\delta
n\rfloor}\right)\end{eqnarray*}

We can now write the second moment by summing up over all possible
couples $\left(\sigma_{1},\sigma_{2}\right)$ :

\begin{eqnarray*}
\mathbf{E}[X_{\delta}^{2}] & = & \sum_{p=0}^{min\left(\lfloor\delta n\rfloor,n-\lfloor\delta n\rfloor\right)}{n \choose \left(\lfloor\delta n\rfloor-p\right)\,\, p\,\, p\,\,\left(n-\lfloor\delta n\rfloor-p\right)}\, G_{I_{k},\delta}\left(\frac{p}{\lfloor\delta n\rfloor}\right)^{rn}\end{eqnarray*}

We now estimate $\mathbf{E}[X_{\delta}^{2}]$ as a function of $n$,
using a classical asymptotic estimate of the multinomial
coefficient. For small multinomial numbers the asymptotic estimate
being also an upper bound, it will be sufficient for the
estimation we need. We set : $\mu=\frac{p}{\lfloor\delta
n\rfloor}$.

\begin{eqnarray*}
{n \choose \left(1-\mu\right)\delta n\,\mu\delta n\,\mu\delta
n\,\left(1-\delta-\mu\delta\right)n} \leq\frac{(2\pi
n)^{-3/2}}{\sqrt{2\mu\delta(1-\mu)\delta(1-\delta-\mu\delta)}}\left(t_{\delta}\left(\mu\right)\right)^{-n}
\end{eqnarray*}

where :
$t_{\delta}\left(\mu\right)=\left(\left(1-\mu\right)\delta\right)^{\left(1-\mu\right)\delta}\left(\mu\delta\right)^{2\mu\delta}\left(1-\delta-\mu\delta\right)^{1-\delta-\mu\delta}$.
We have :

\begin{eqnarray}
\label{esp1} \mathbf{E}[X_{\delta}^{2}]\leq\sum_{\mu \in
\{0,\frac{1}{\lfloor\delta
n\rfloor},\ldots,min\left(1,\frac{n-\lfloor\delta
n\rfloor}{\lfloor\delta n\rfloor}\right) \}}\frac{(2\pi
n)^{-3/2}}{\sqrt{2
\mu\delta(1-\mu)\delta(1-\delta-\mu\delta)}}\left(\Gamma_{I_k,\delta,r}\left(\mu\right)\right)^{n}
\end{eqnarray}

with:

\begin{equation}
\label{gamma}
\Gamma_{I_{k},\delta,r}\left(\mu\right)=\frac{G_{I_k,\delta}\left(\mu\right)^{r}}{t_{\delta}\left(\mu\right)}
\end{equation}

Each term in (\ref{esp1}) consisting of a polynomial factor and an
exponential factor in $n$ the sum can be estimated with a discrete
version of Laplace method. Thus :

\begin{lem}
\label{laplace}if $u$ and $v$ are smooth real-valued functions of
one variable $x$, and if $v$ has a single maximum on
$[a,+\infty[,$ located at $\xi_{0}$with $\xi>a$ and if further
$v^{\prime\prime}(\xi_{0})$ is not 0, then :

\[
\frac{1}{\sqrt{n}}\sum_{i=a\omega n}^{\omega
n}u(\frac{i}{n})exp(nv(\frac{i}{n}))\sim
g(\xi_{0})\frac{\sqrt{2\pi}\omega}{\sqrt{|v^{\prime\prime}(\xi_{0})|}}exp(nv(\xi_{0}))\]

\end{lem}

We will apply Lemma \ref{laplace}, setting
$v(\mu)=\log(\Gamma_{I_{k},\delta,r}\left(\mu\right))$ and
$u(\mu)=\frac{(2\pi)^{-\frac{3}{2}}}{\sqrt{2\mu\delta^{2}(1-\mu)(1-\delta-\mu\delta)}}$.
Then :
\begin{equation}
\label{lap} \mathbf{E}[X_{\delta}^{2}]\leq
n^{-1}\frac{(2\pi)^{-3/2}}{\sqrt{2(1-\delta)^{3}\delta^{3}}}\times\frac{\sqrt{2\pi}min\left(\delta,(1-\delta)\right)}{\sqrt{|\Gamma^{\prime\prime}_{I_{k},\delta,r}(1-\delta)|}}\left(\max_{\mu
\in[0,min\left(1,\frac{1-\delta}{\delta}\right)]}{\Gamma_{I_{k},\delta,r}\left(\mu\right)}\right)^{n}
\end{equation}

The success of the second moment method relies mainly on the behavior
of $\Gamma_{I_{k},\delta,r}\left(\mu\right)$ for $\mu\in[0,min\left(1,\frac{1-\delta}{\delta}\right)]$

The upper bound of the ratio $\frac{\mathbf{E}[X_{\delta}]^{2}}{\mathbf{E}[X_{\delta}^{2}]}$ will then  depend mainly on its exponential part $\frac{\gamma_{I_k,r}\left(\delta\right)^2}{\max_{\mu \in[0,min\left(1,\frac{1-\delta}{\delta}\right)]}{\Gamma_{I_{k},\delta,r}\left(\mu\right)}}$ that must be equal to $1$ otherwise, all what we will get is the trivial relation $\frac{\mathbf{E}[X_{\delta}]^{2}}{\mathbf{E}[X_{\delta}^{2}]}\geq0$. We will see in the next Section that this achieved through characteristic solutions.

From (\ref{phi}) we will write in the sequel of the paper :

\begin{eqnarray}
 \phi_{i,j,\delta}\left(\mu\right) &
= & \sum_{d=max(0,i+j-k)}^{min(i,j)}{k \choose i}{i \choose d}{k-i
\choose j-d}\kappa_{i,j,d,\delta}\left(\mu\right)\end{eqnarray}

\begin{eqnarray}
\label{tdelta}
\text{with :\hspace{0.5cm}}
\kappa_{i,j,d,\delta}\left(\mu\right)=\left(\left(1-\mu\right)\delta\right)^{d}\left(\mu\delta\right)^{i+j-2d}\left(\left(1-\delta-\mu\delta\right)\right)^{k-i-j+d}
\end{eqnarray}
and :
\begin{eqnarray}
\label{GIK}
G_{I_{k},\delta}\left(\mu\right) & =
& \sum_{i\in I_{k}}\sum_{j\in
I_{k}}\phi_{i,j,\delta}\left(\mu\right)\end{eqnarray}

\subsection{Proof of Theorem \ref{lem:secondmoment} and its Corollary}
 \label{preuve13}
In the following, we sketch first the proof by discussing its most important ingredients. A crucial point for the success of the method is the point where $\mu=1-\delta$ or the independence point. To understand this, consider two valuations drawn independently uniformly at random from the set of $\delta$-valuations. A variable is assigned $1$ under one of the two $\delta$-valuations with probability $\delta$ and $0$ with probability $1-\delta$. Since the  two valuations are selected independently, the probability of being assigned $1$ by a $\delta$-valuation and $0$ by the other is $\delta(1-\delta)$. Thus, according to the notation of the Section \ref{sub:Second-moment-method}, these pairs of  $\delta$-valuations are characterized by Hamming distance $2\mu$ with $\mu=1-\delta$.  These uncorrelated pairs of $\delta$-valuations play a central role in the success of the method.  Indeed:

\begin{eqnarray*}
G_{I_{k},\delta}\left(1-\delta\right) &=&\sum_{i\in I_{k}}\sum_{j\in I_{k}}\phi_{i,j,\delta}\left(1-\delta\right)\\
&=& \sum_{i\in I_{k}}\sum_{j\in I_{k}} \delta^{i+j}\left(1-\delta\right)^{2k-i-j}{k \choose i}\sum_{d}{i \choose d}{k-i \choose j-d}\\
& =&\sum_{i\in I_{k}}\sum_{j\in I_{k}} \delta^{i+j}\left(1-\delta\right)^{2k-i-j}{k \choose i} {k \choose j} \\
&=&\sum_{i\in I_{k}}\sum_{j\in I_{k}} \pi_i\left(\delta\right)\pi_{j}\left(\delta\right)\\
&=&g_{I}\left(\delta\right)^{2}
\end{eqnarray*}

 It is easy to see that $t_{\delta}\left(1-\delta\right)=\left(\delta^{\delta}\left(1-\delta\right)^{1-\delta}\right)^{2}$, thus :
\begin{equation}
 \Gamma_{I_{k},\delta,r}\left(1-\delta\right)=\frac{G_{I_k,\delta}\left(1-\delta\right)^{r}}{t_{\delta}\left(1-\delta\right)}=\frac{g_{I_k}\left(\delta\right)^{2r}}{\left(\delta^{\delta}\left(1-\delta\right)^{1-\delta}\right)^{2}}=\gamma_{I_{k},r}\left(\delta\right)^{2}
 \label{SMC}
 \end{equation}

Since $\Gamma_{I_{k},\delta,r}\left(1-\delta\right)/\gamma_{I_{k},r}\left(\delta\right)^{2}=1$,
if $\mu=1-\delta$ is not the global maximum of $\Gamma_{I_{k},\delta,r}\left(\mu\right)$, there exist some $\mu$ for which $\gamma_{I_{k}}\left(\delta\right)^{2}/\Gamma_{I_{k},\delta,r}\left(\mu\right)<1$ making the method to fail.

Consequently, a necessary condition for the success of the method is that $\mu=1-\delta$ is a stationary point. Lemma \ref{lem:PTstat} states that this is the case only for the characteristic solutions. 

Let $\rho=\max_{\mu \in [0,min(1,\frac{1-\delta}{\delta})]}\left(\frac{\delta}{1-\mu}+\frac{2\delta}{\mu}+\frac{\delta^{2}}{1-\delta-\delta\mu}\right)$, $\nu=\max_{\mu \in [0,min(1,\frac{1-\delta}{\delta})]}\left(\log \left(G_{I_{k},\delta}\left(\mu\right)\right)\right)^{\prime\prime}$. Let:
\begin{equation}
\label{rik}
r^*_{I_k} =\frac{\rho}{\nu} 
\end{equation}
The Lemma \ref{lem:maxglobal} states that  $r^*_{I_k}$ is well defined and that it is strictly positive and for any $r<r^*_{I_k}$,  the second derivative of $\log \left(\Gamma_{I_{k},\delta,r}\left(\mu\right)\right)$ is negative for any $\mu \in [0,min\left(1,(1-\delta)/\delta \right)]$. This function is then concave  in the previous interval. Combining the two lemmas, we can conclude that there is a range of ratios $]0,r^*_{I_k}[$ for which $\mu=1-\delta$ is the global maximum of $\Gamma_{I_k,\delta,r}$. 

\begin{remark}
In general, the point $\mu=1-\delta$ continues to be the global maximum in a range beyond $r^*_{I_k}$ after the function ceases to be concave, allowing through a more precise analysis to get better lower bound than $r^*_{I_k}$. However, a general bound beyond concavity is hard to figure out for the class and we do not need this fact for the proof of Theorem \ref{lem:secondmoment} which aim is to give the conditions under which the second moment succeeds regardless of the value of the bound obtained. When one needs for a particular problem to compute the best lower bound with respect to $\delta$-solutions, a finer analysis is required for this particular problem. This is what we do to get the best possible lower bound with respect to $\delta$-solutions for positive $1$-in-$k$-SAT.
\end{remark}

\begin{lem}
\label{lem:PTstat}  $\Gamma_{I_{k},\delta,r}^\prime\left(1-\delta\right)=0$ iff $\delta\in\Delta_{I_{k}}$. \end{lem}
\begin{proof}
Considering (\ref{gamma}), it is easy to check that the derivative of $t_\delta^\prime\left( \mu\right)$ (defined in (\ref{tdelta})) is such that $t_\delta^\prime\left( 1-\delta\right)=0$. It is then necessary and sufficient  that $G_{I_k,\delta}^\prime\left(1-\delta\right)=0$.  It can be shown (see Appendix \ref{eqgik}), that:
\begin{equation}
\label{giketg}
G_{I_k,\delta}^\prime\left(1-\delta\right)=-\frac{\delta \left({g}_{I_k}\left(\delta\right)^\prime\right)^2}{k \left(1-\delta\right)^2}
\end{equation}
which is equal to $0$ iff ${g}_{I_k}\left(\delta\right)^\prime=0$ i.e. $\delta \in \Delta_{I_k}$.
\end{proof}
\begin{lem}
\label{lem:maxglobal} $r^*_{I_k}$ (as defined in (\ref{rik})) is strictly greater than $0$  and for 
every $r<r^*_{I_k}$,   $\log \left(\Gamma_{I_{k},\delta,r}\left(\mu\right)\right)^{\prime \prime}<0$ for $\mu \in [0,min\left(1,(1-\delta)/\delta \right)]$. \end{lem}
\begin{proof}
\begin{eqnarray*}
\left(\log\left(\Gamma_{I_{k},\delta,r}\left(\mu\right)\right)\right)^{\prime\prime} &=& \left(-\frac{\delta}{1-\mu}-\frac{2\delta}{\mu}-\frac{\delta^{2}}{1-\delta-\delta\mu}\right)+r \left(\log\left(G_{I_{k},\delta}\left(\mu\right)\right)\right)^{\prime\prime}\\
\end{eqnarray*}
It can be shown (see Appendix \ref{conc}) that $\left(-\frac{\delta}{1-\mu}-\frac{2\delta}{\mu}-\frac{\delta^{2}}{1-\delta-\delta\mu}\right)$ is negative and bounded from above by $-\rho$ and that $\left(\log\left(G_{I_{k},\delta}\left(\mu\right)\right)\right)^{\prime\prime}$ is positive and bounded from above by $\nu$, then $\left(\log\left(\Gamma_{I_{k},\delta,r}\left(\mu\right)\right)\right)^{\prime\prime}  \leq -\rho+r \nu<0$ if $r<{\rho}/{\nu}=r^*_{I_k}$.
\end{proof}

Now we are in position to give the proof of Theorem \ref{lem:secondmoment}.

\textbf{Proof of Thorem \ref{lem:secondmoment}}

Thanks to Lemma  \ref{lem:PTstat}, we know that $\mu=1-\delta$ is a stationary point for $ \log(\Gamma_{I_{k},\delta,r}\left(\mu\right))$ and thanks to Lemma \ref{lem:maxglobal}, we know that  $\log(\Gamma_{I_{k},\delta,r}\left(\mu\right))$ is concave for $r<r^*_{I_k}$. Combining these two facts, we deduce that $\mu=1-\delta$ is a global maximum for $\log(\Gamma_{I_{k},\delta,r}\left(\mu\right))$. The inequality (\ref{lap}) becomes:

\[
\mathbf{E}[X_{\delta}^{2}]\leq n^{-1}\frac{(2\pi)^{-3/2}}{\sqrt{2(1-\delta)^{3}\delta^{3}}}\times\frac{\sqrt{2\pi}min\left(\delta,(1-\delta)\right)}{\sqrt{|\Gamma^{\prime\prime}_{I_{k},\delta,r}(1-\delta)|}}(\Gamma_{I_{k},\delta,r}\left(1-\delta\right))^{n}\]

On putting $C_{1}=\frac{(2\pi)^{-3/2}}{\sqrt{2(1-\delta)^{3}\delta^{3}}}\times\frac{\sqrt{2\pi}min\left(\delta,(1-\delta)\right)}{\sqrt{|\Gamma^{\prime\prime}_{I_{k},\delta,r}(1-\delta)|}}$
and having thanks to (\ref{SMC}) $\Gamma_{I_{k},\delta,r}\left(1-\delta\right)=\gamma_{I_{k},r}(\delta)^{2}$,
this yields : $\mathbf{E}[X_{\delta}^{2}]\leq n^{-1}C_{1}(\gamma_{I_{k},r}(\delta)^{2n})$.
Allowing for the relation $\mathbf{E}[X_{\delta}]\geq\frac{e^{-1/6}}{\sqrt{2\pi\delta(1-\delta)n}}\left(\gamma_{I_k,r}\left(\delta\right)\right)^{n}$,
we deduce :

\[
\frac{\mathbf{E}[X_{\delta}]^{2}}{\mathbf{E}[X_{\delta}^{2}]}\geq C_{2}\]

$C_{2}$ being a positive constant.  Thus Theorem \ref{lem:secondmoment} is proved. 

From the Creignou and Daud\'e criterion \cite{CombinatorialsharpnesscriterionandphasetransitionclassificationforrandomCSPs.::2004::CreignouDaude,DBLP:journals/dm/CreignouD09}
for $k\geq3$ a $CSP(I_{k})$ is neither depending on one component nor
strongly depending on a 2XOR-relation, it can be stated according to
the Friedgut's theorem in \cite{Sharpthresholdsforgraphpropertiesandthek-SATproblem::1999::Friedgut}
the following fact :
\begin{fact}
For every $k\geq3$ and a random $CSP(I_{k})$, there exists a function
$\lambda_{k}(n)$ such that for any $\epsilon>0$:

\[
\lim_{n\rightarrow\infty}Pr(I_{k}(rn,n)\: is\: satisfiable)=\begin{cases}
1 & \mbox{if}\: r>\lambda_{k}(n)(1-\epsilon)\\
0 & \mbox{if}\: r<\lambda_{k}(n)(1-\epsilon)\end{cases}\]

\end{fact}
It follows that for any $CSP\left(I_{k}\right),$ if $r<r^*_{I_k}$
as defined (\ref{rik}), then : 

$\lim{}_{n\rightarrow\infty}\Pr\left(I_{k}\left(n,rn\right)\mbox{ is satisfiable }\right)=1$.
Thus Corollary \ref{thm:borneinf} is proved.

\section{\label{sec:Monotone--in--SAT-case}Positive $1$-in-$k$-SAT case:
proof of Theorem \ref{thm:1ksat}}

For $1$-in-$k$-SAT, we denote the corresponding $I_{k}$ by $1_{k}=\{1\}$.
The function $g_{1_{k}}\left(\delta\right)=k\delta\left(1-\delta\right)^{k-1}$.
It is easy to check that $\Delta_{1_{k}}=\{\frac{1}{k}\}$. 
We note first that the best lower bound that we can hope to get is $\hat{r}_{1_k,1/k}$ as defined in Remark \ref{best}. It is easy to check that $\lim_{k\rightarrow \infty} \frac{\hat{r}_{1_k,1/k}}{\log k/k}=1$. Then asymptotically, the best lower bound that can be obtained with respect to $1/k$-solutions is  $\log k/k$.

For the
second moment, as previously we consider only $\delta$-valuations.
Only the function $G_{1_{k},1/k}\left(\mu\right)$ changes:
\[
G_{1_{k},1/k}\left(\mu\right)=k\left(k-1\right)\kappa_{1,1,1,1/k}\left(\mu\right)+k\kappa_{1,1,0,1/k}\left(\mu\right)=k^{1-k}\left(k-1-\mu\right)^{k-2}\left(k\left(1-\mu+\mu^{2}\right)-1\right)\]
Thanks to Lemma $ $\ref{lem:PTstat}, we know that $\mu=1-1/k$ is
stationary point of $\Gamma_{1_{k},\frac{1}{k},r}\left(\mu\right)$
and that $\Gamma_{1_{k},\frac{1}{k},r}\left(1-1/k\right)=\gamma_{1_{k},r}\left(1/k\right)^{2}$.
It remains to prove that it is a global maximum for $r\leq\log k/k$. This bound goes in general beyond concavity so we need a finer analysis. That is the purpose of this Lemma.
\begin{lemma}
\label{maxim}
For any $\mu \in [0,1]$, $\Gamma_{1_{k},\frac{1}{k},r}\left(\mu\right) \leq \gamma_{1_{k},r}\left(1/k\right)^{2}$.
\end{lemma}
 \begin{proof}
We give here just an outline of the proof. A detailed one is given in the Appendix \ref{detail_maxim}. 

The interval of $\mu$ is divided into two parts $[0,1/2]$ and $[1/2,1]$ where the function $\Gamma_{1_{k},\frac{1}{k},r}$ is bounded from above using two different techniques.

 First for $\mu \in [0,1/2]$: In this interval, we use mainly the fact that for some $a \in [0,1/2]$, $1-\mu-\mu^2\leq l_a=1-a-a^2$ for any $\mu \in [a,1/2]$. Let $$\tau_{a}\left(\mu\right)=\log k\,\log\left(k^{1-k}\left(k-1-\mu\right)^{k-2}\left(k l_a-1\right)\right)-k\log t_{1/k}\left(\mu\right)$$ $\tau_{a}\left(\mu\right)$ bounds from above $k\,\log\Gamma_{1_{k},\frac{1}{k},\frac{\log k}{k}}\left(\mu\right)$ in the interval $[a,1/2]$. It can be shown that $\tau_{a}\left(\mu\right)$ is strictly increasing in the above interval. Beginning with $a_0=0$, we find a value $a_1$ such that $\tau_{a_0}\left(a_1\right) < 2 k \log \gamma_{1_{k},r}\left(1/k\right)$ proving the desired inequality for $\mu \in [a_0,a_1]$. We repeat the same with $\tau_{a_1}$ and find a $a_2$ and so on... until an $a_i \geq 1/2$ which finishes this part of the proof. In fact, only two steps are sufficient with $a_1=0.15$.

 Second for $\mu \in [1/2,1]$: Recall that $k\log\Gamma_{1_{k},\frac{1}{k},\frac{\log k}{k}}\left(\mu\right)=\log k\,\log G_{1_{k},\frac{1}{k}}\left(\mu\right)-k\log t_{1/k}\left(\mu\right)$. We prove first separately that the derivatives of $\log G_{1_{k},\frac{1}{k}}\left(\mu\right)$ and $-\log t_{1/k}\left(\mu\right)$ are concave in the whole considered interval. Then we split the above interval in two parts $]1/2,1-1/k[$ and $]1-1/k,1]$. Considering their concavity, both functions can be bound from below in the first interval by the linear functions representing the chords joining the two points corresponding to the two bounds of the interval. The sum of these two linear functions being positive, this proves that the derivative of $k\log\Gamma_{1_{k},\frac{1}{k},\frac{\log k}{k}}\left(\mu\right)$ is positive in the first interval and then that the value of the function at $\mu=1-1/k$ is maximum. For the second interval, the functions are bounded from above by the linear functions representing the tangent lines at $\mu=1-1/k$. The sum of these two linear functions being negative, the derivative of $k\log\Gamma_{1_{k},\frac{1}{k},\frac{\log k}{k}}\left(\mu\right)$ is negative in the second interval and then $\mu=1-1/k$ is also the maximum in the second interval. Summing up, $\Gamma_{1_{k},\frac{1}{k},\frac{\log k}{k}}\left(1-1/k\right)=\gamma_{1_{k},r}\left(1/k\right)^{2}$ is the maximum of $\Gamma_{1_{k},\frac{1}{k},\frac{\log k}{k}}\left(\mu\right)$ within  $[1/2,1]$.

 \end{proof}
\subsection{A general upper bound for positive $1$-in-$k$-SAT}

$X$ is the random variable associating to each $1_{k}(m,n)$ the
number of its solutions.. We have:

\[
\mathbf{E}X=\sum_{p=0}^{n}{n \choose p}\left(k\frac{p}{n}\left(1-\frac{p}{n}\right)^{k-1}\right)^{rn}\sim\left(\max_{\delta\in[0,1]}\left(\gamma_{1_{k},r}\left(\delta\right)\right)\right)^{n}poly(n)\]

For the upper bound, we prove that for $k\geq7$ and $r=\log^{2}k/k$,
$\max_{\delta\in[0,1]}\left(\gamma_{1_{k}}\left(\delta\right)\right)<1$.
This is the purpose of the following Fact. 
\begin{fact}
for $k\geq7$ , $\max_{\delta\in[0,1]}\left(\gamma_{1_{k},\log^{2}k/k}\left(\delta\right)\right)<1$.\end{fact}
\begin{proof}
We prove it first in the interval $\delta\in[1/2,1]$. Both $g_{1_{k}}\left(\delta\right)^{r}$
and $\frac{1}{\delta^{\delta}\left(1-\delta\right)^{1-\delta}}$ are
decreasing in $\delta$ in this interval. $\gamma_{1_{k},r}\left(1/2\right)=2g_{1_{k}}\left(1/2\right)^{r}=2\left(\frac{k}{2^{k}}\right)^{\log^{2}k/k}<1$
then $\gamma_{1_{k},\log^{2}k/k}\left(\delta\right)<1$ in the interval
$\delta\in[1/2,1]$ . $g_{1_{k}}\left(\delta\right)=k\delta\left(1-\delta\right)^{k-1}$
increases from $0$ until $\delta=1/k$. In the same interval $\delta^{\delta}\left(1-\delta\right)^{1-\delta}$
is decreasing. Then $\gamma_{1_{k},\log^{2}k/k}\left(\delta\right)=\frac{g_{1_{k}}\left(\delta\right)^{\log^{2}k/k}}{\delta^{\delta}\left(1-\delta\right)^{1-\delta}}\leq\frac{g_{1_{k}}\left(1/k\right)^{\log^{2}k/k}}{\left(1-1/k\right)^{1-1/k}\left(1/k\right)^{1/k}}\leq ke^{\frac{\text{(-1+k) (-1+\ensuremath{\log}k) (1+\ensuremath{\log}k})}{k}}<1$
for $k\geq7$. 

It remains to handle the function within the interval $[1/k,1/2]$.
Since $\log\frac{1}{\delta^{\delta}\left(1-\delta\right)^{1-\delta}}$
is concave it can be bound from above by the line of slope its derivative
$\frac{1}{\delta^{\delta}\left(1-\delta\right)^{1-\delta}}\leq(-1+k)^{-1+\delta}k$.
$\frac{g_{1_{k}}\left(\delta\right)^{\log^{2}k/k}}{\delta^{\delta}\left(1-\delta\right)^{1-\delta}}\leq(-1+k)^{-1+\delta}k\, g_{1_{k}}\left(\delta\right)^{\log^{2}k/k}$.
$(-1+k)^{-1+\delta}k\, g_{1_{k}}\left(1/k\right)^{\log^{2}k/k}$ is
less than $1$ within $[1/k,s]$ where $s=-\log\left(1-1/k\right)\left(\left(k-1\right)\log\left(k\right)^{2}-k\right)/(k\log\left(k-1\right))$.
Finally, we bound from above $\log g_{1_{k}}\left(\delta\right)$
by $\log g_{1_{k}}^{\prime}\left(s\right)\left(\delta-s\right)+\log g_{1_{k}}\left(s\right)$.
The upper bound is less than $1$ for $k>7$ in $[s,1/2]$. 
\end{proof}
The application of Markov inequality finishes the proof of the upper
bound.

\section{\label{sec:Discussion-and-conclusion}Discussion}


Any element of $\Delta_{I_k}$ is necessary and sufficient to make the second moment method to be successful as stated by theorem \ref{lem:secondmoment}.  An interesting question that is raised by the fact that $\Delta_{I_k}$ may have many values is : what value gives the better lower bound? Precisely, is there a simple criterion that permits to select the $\delta \in \Delta_{I_k}$ that gives the best lower bound? 

In the example of Figure \ref{getPM}, the function $g_{I_{13}}$ is represented for  $I_{13}=\{1,8,12 \}$. It has three local maxima and so $\Delta_{I_{13}}=\{\delta_1,\delta_2,\delta_3\}$ with $g_{I_{13}}(\delta_2) < g_{I_{13}}(\delta_1) < g_{I_{13}}(\delta_3)$. As said before, the second moment method succeeds only for those three values of $\delta$. An immediate candidate for this choice of the best value could be $\delta_3$ since it is the one for which the probability of satisfying a randomly selected constraint is maximum.  In fact, the best lower bound is obtained using $\delta_2$. The latter is the one that maximizes the first moment of $X_\delta$ i.e. that corresponds to $max_{\delta \in \Delta_{I_k}} \left (\gamma_{I_k}(\delta)\right)$. Since $\gamma_{I_k}(\delta)=\delta^{-\delta} \left( 1-\delta \right)^{\delta-1} g_{I_{k}}(\delta)$, the entropy term $\delta^{-\delta} \left( 1-\delta \right)^{\delta-1}$ centered on $1/2$ tends to favor values of $\delta$ near $1/2$.  We have verified this fact for many problems. We conjecture that for any problem defined by  the set $I_k$, the best value of $\delta$ for the second moment method is the $\delta^* \in \Delta_{I_k}$ that maximizes $\gamma_{I_k}(\delta)$. 

\begin{figure}[h]
\includegraphics[width=8cm]{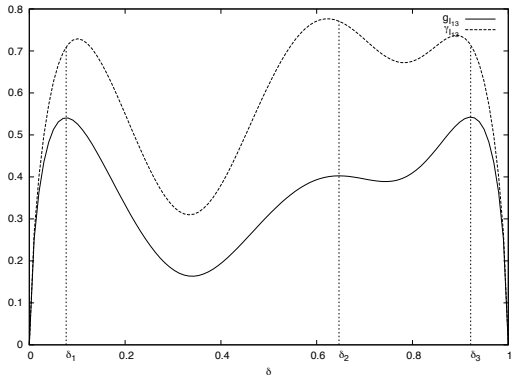}
\caption{An example of the functions $g_{I_k}$ and $\gamma_{I_k}$ for $I_{13}=\{1,8,12 \}$ for $r=0.64$. }
\label{getPM}
\end{figure}


\bibliographystyle{plain}
\bibliography{library}

\newpage
\begin{appendix}

\section{Proof of the equality (\ref{giketg})}
\label{eqgik}
\begin{proof}
Considering (\ref{GIK}) :
\begin{eqnarray*}
G_{I_k,\delta}^\prime\left(1-\delta\right)&=&\sum_{i\in I_{k}}\sum_{j\in I_{k}}\phi_{i,j,\delta}^\prime\left(1-\delta\right)\\
&=&\sum_{i\in I_{k}}\sum_{j\in I_{k}}\sum_{d}{i \choose d}{k-i \choose j-d} \kappa_{i,j,d,\delta}^{\prime}\left(1-\delta\right)\\
\end{eqnarray*}
Recall that $\kappa_{i,j,d,\delta}\left(\mu\right)=\left(\left(1-\mu\right)\delta\right)^{d}\left(\mu\delta\right)^{i+j-2d}\left(\left(1-\delta-\mu\delta\right)\right)^{k-i-j+d}$ and then
\begin{equation*}
\kappa_{i,j,d,\delta}^{\prime}\left(\mu\right)=\kappa_{i,j,d,\delta}\left(\mu\right)\left(-\frac{d}{1-\mu}+\frac{i+j-2d}{\mu}-\frac{\delta\left(k-i-j+d\right)}{\left(1-\delta-\delta\mu\right)}\right)\end{equation*} 

Noting that $\kappa_{i,j,d,\delta}\left(1-\delta\right)=\delta^{i+j}\left(1-\delta\right)^{2k-i-j}$ 
we get:

$\phi_{i,j,\delta}^{\prime}\left(1-\delta\right)=\delta^{i+j}\left(1-\delta\right)^{2k-i-j}{k \choose i}\sum_{d}{i \choose d}{k-i \choose j-d}\left(-\frac{d}{\delta}+\frac{i+j-2d}{1-\delta}-\frac{\delta\left(k-i-j+d\right)}{\left(1-\delta\right)^{2}}\right)$.

Using the mean of the hypergeometric distribution of parameters $k$, $i$ and $j$ ($\sum_{d=0}^{j}d{i \choose d}{k-i \choose j-d}/{k \choose j}=\frac{ij}{k}$)
and Vandermonde identity, we get:

$\phi_{i,j,\delta}^{\prime}\left(1-\delta\right)=\delta^{i+j}\left(1-\delta\right)^{2k-i-j}{k \choose i}{k \choose j}\left(-\frac{ij}{k\delta}+\frac{i+j-2ij/k}{1-\delta}-\frac{\delta\left(k-i-j+ij/k\right)}{\left(1-\delta\right)^{2}}\right)$. 

Denoting the quantity $h_{I_{k}}\left(\delta\right)=\sum_{i\in I}{\displaystyle i{k \choose i}\delta^{i}\left(1-\delta\right)^{k-i}}$
\begin{eqnarray*}
\sum_{i\in I_{k}}\sum_{j\in I_{k}}\delta^{i+j}\left(1-\delta\right)^{2k-i-j}{k \choose i}\sum_{d}{i \choose d}{k-i \choose j-d} i j&=&h_{I_{k}}\left(\delta\right)^{2}\\
\sum_{i\in I_{k}}\sum_{j\in I_{k}}\delta^{i+j}\left(1-\delta\right)^{2k-i-j}{k \choose i}\sum_{d}{i \choose d}{k-i \choose j-d} i j&=&h_{I_{k}}\left(\delta\right) g_{I_{k}}\left(\delta\right)\\
\sum_{i\in I_{k}}\sum_{j\in I_{k}}\delta^{i+j}\left(1-\delta\right)^{2k-i-j}{k \choose i}\sum_{d}{i \choose d}{k-i \choose j-d} &=&g_{I_{k}}\left(\delta\right)^{2}
\end{eqnarray*}

%
%
%
%
%
%
%
%

\begin{eqnarray*}
 G_{I_{k},\delta}^{\prime}\left(1-\delta\right)&=&\sum_{i\in I_{k}}\sum_{j\in I_{k}}\delta^{i+j}\left(1-\delta\right)^{2k-i-j}{k \choose i}{k \choose j}\left(-\frac{ij}{k\delta}+\frac{i+j-2ij/k}{1-\delta}-\frac{\delta\left(k-i-j+ij/k\right)}{\left(1-\delta\right)^{2}}\right)\\
 &=&-\frac{h_{I_{k}}\left(\delta\right)^{2}}{k\delta}+\frac{2h_{I_{k}}\left(\delta\right)g_{I_{k}}\left(\delta\right)-2 h_{I_{k}}\left(\delta\right)^{2}/k}{1-\delta}\\
 & &-\frac{\delta\left(kg_{I_{k}}^{2}\left(\delta\right)-2h_{I_{k}}\left(\delta\right)g_{I_{k}}\left(\delta\right)+h_{I_{k}}\left(\delta\right)^{2}/k\right)}{\left(1-\delta\right)^{2}}\\
 &=&-\frac{\left(h_{I_k}(\delta)-k \delta g_{I_k}(\delta)\right)^2}{\delta(1-\delta)^2}
 \end{eqnarray*}

Noting that because of:
\begin{eqnarray*}
g_{I_{k}}^{\prime}\left(\delta\right)&=&\sum_{i\in I}{k \choose i}\,\delta^{i}\left(1-\delta\right)^{k-i}\left(\frac{i}{\delta}-\frac{k-i}{1-\delta}\right)=\frac{1}{\delta}h_{I_k}(\delta)-k\, g_{I_{k}}\left(\delta\right)
\end{eqnarray*}
$h_{I_k}(\delta)-k \delta g_{I_k}(\delta)=\delta g_{I_{k}}^{\prime}\left(\delta\right)$ allowing for the desired relation.

\end{proof}

\section{Proof of Lemma \ref{lem:maxglobal}}
\label{conc}
\begin{proof}
The  second derivative of $\log\left(\Gamma_{I_{k},\delta,r}\left(\mu\right)\right)$
is: 

\begin{eqnarray*}
\left(\log\left(\Gamma_{I_{k},\delta,r}\left(\mu\right)\right)\right)^{\prime\prime} &=& \left(-\frac{\delta}{1-\mu}-\frac{2\delta}{\mu}-\frac{\delta^{2}}{1-\delta-\delta\mu}\right)+r \left(\log\left(G_{I_{k},\delta}\left(\mu\right)\right)\right)^{\prime\prime}\\
&=&\left(-\frac{\delta}{1-\mu}-\frac{2\delta}{\mu}-\frac{\delta^{2}}{1-\delta-\delta\mu}\right)+r\frac{G_{I_{k},\delta}\left(\mu\right)G_{I_{k},\delta}^{\prime\prime}\left(\mu\right)-G_{I_{k},\delta}^{'}\left(\mu\right)^{2}}{G_{I_{k},\delta}\left(\mu\right)^{2}}
\end{eqnarray*}

It is easy to check that the second derivative of $-\frac{\delta}{1-\mu}-\frac{2\delta}{\mu}-\frac{\delta^{2}}{1-\delta-\delta\mu}$ is negative and that its derivative tends to $\infty$ when $\mu$ tends to $0$ and to $-\infty$ on the other side then $-\frac{\delta}{1-\mu}-\frac{2\delta}{\mu}-\frac{\delta^{2}}{1-\delta-\delta\mu}$
increases from $-\infty$ attains a maximum at a negative value
then decreases to $-\infty$. Let $-\rho$ ($\rho>0$) be its maximum
value.

In the second part $G_{I_{k},\delta}\left(\mu\right)^{2}$ is bounded
and strictly positive. Indeed it is formed by a sum of positive terms
some of which are strictly positive. Indeed all $\kappa_{i,j,d,\delta}\left(\mu\right)>0$
for every $\mu\in]0,min\left(1,\frac{1-\delta}{\delta}\right)[$.
Moreover, $ $$\kappa_{i,i,i,\delta}\left(0\right)>0$ and if $\delta\leq1/2$
$\kappa_{i,i,0,\delta}\left(1\right)>0$ otherwise $ $$\kappa_{i,i,2i-k,\delta}\left(\left(1-\delta\right)/\delta\right)>0$.

$G_{I_{k},\delta}\left(\mu\right)G_{I_{k},\delta}^{\prime\prime}\left(\mu\right)-G_{I_{k},\delta}^{'}\left(\mu\right)^{2}$
is a polynomial in $\mu$. It is also bounded for $\mu\in[0,min\left(1,\frac{1-\delta}{\delta}\right)]$.
The second part have no singular point and it it bounded. Le $\nu$ be its maximum value. We prove that $\nu>0$. We know thanks
to Lemma \ref{lem:PTstat} that $G_{I_{k},\delta}^{\prime}\left(1-\delta\right)=0$.
Moreover the second derivative $G_{I_{k},\delta}^{\prime\prime}\left(1-\delta\right)=\frac{\delta^{2}g_{I_k}^{\prime\prime}\left(\delta\right)^{2}}{k\left(k-1\right)}$
and as seen before $G_{I_{k},\delta}\left(1-\delta\right)=g_{I_k}\left(\delta\right)^2$. We deduce
that $\nu>0$. Indeed, $$\nu\geq \frac{G_{I_{k},\delta}\left(1-\delta\right)G_{I_{k},\delta}^{\prime\prime}\left(1-\delta\right)-G_{I_{k},\delta}^{'}\left(1-\delta\right)^{2}}{G_{I_{k},\delta}\left(1-\delta\right)^2}=\frac{\delta^{2}g_{I_k}^{\prime\prime}\left(\delta\right)^{2}}{k\left(k-1\right)g_{I_k}\left(\delta\right)^{2}}>0$$.
Finally :

\[
\left(-\frac{\delta}{1-\mu}-\frac{2\delta}{\mu}-\frac{\delta^{2}}{1-\delta-\delta\mu}\right)+r\frac{G_{I_{k},\delta}\left(\mu\right)G_{I_{k},\delta}^{\prime\prime}\left(\mu\right)-G_{I_{k},\delta}^{'}\left(\mu\right)^{2}}{G_{I_{k},\delta}\left(\mu\right)^{2}}\leq-\rho+r.\nu\]

The second derivative if then negative over $[0,min\left(1,\frac{1-\delta}{\delta}\right)]$
for every $r<r^*_{I_k}=\rho/\nu$.
\end{proof}
\section{Detailed proof of Lemma \ref{maxim}}
\label{detail_maxim}
\begin{proof}

${\bf \mu\in[0,1/2]}$: 
For $\mu\in[0,1]$he second derivative of $-\log t_{1/k}\left(\mu\right)$
is negative and so is the second derivative of $\log\left(k-1-\mu\right)$.
This permits to conclude that $\tau_{a}^{\prime}\left(\mu\right)$
is decreasing. $\tau_{a}^{\prime}\left(1/2\right)=\log(2\, k-3)-(2(k-2)\log\left(k\right))/(2k-3)>0$
for every $k>3$ . So $\tau_{a}^{\prime}\left(\mu\right)>0$ for $\mu\in[0,1/2]$
and then $\tau_{a}\left(\mu\right)$ is strictly increasing in the
same interval. 

It is easy to check that $\tau_{0}\left(0.15\right)-2\, k\,\log\left(\gamma_{1_{k},\frac{\log k}{k}}\left(\frac{1}{k}\right)\right)<0$
for every $k>3$. Consequently  $\log \Gamma_{1_{k},\frac{1}{k},\frac{\log k}{k}}\left(\mu\right)<2\,\log\left(\gamma_{1_{k},\frac{\log k}{k}}\left(\frac{1}{k}\right)\right)$ for every $\mu \in [0,0.15]$.
 
Similarly  $\tau_{0.15}\left(0.5\right)-2\, k\,\log\left(\gamma_{1_{k},\frac{\log k}{k}}\left(1/k\right)\right)<0$ for any $k>3$. Concluding that $$\log \Gamma_{1_{k},\frac{1}{k},\frac{\log k}{k}}\left(\mu\right)<2\,\log\left(\gamma_{1_{k},\frac{\log k}{k}}\left(\frac{1}{k}\right)\right)$$ in the interval $[0,\frac{1}{2}]$.

%

${\bf \mu\in[1/2,1[}$:  The second derivative of $-k\log t_{1/k}\left(\mu\right)$
is $-\frac{1}{1-\mu}-\frac{1}{k-1-\mu}-\frac{2}{\mu}$. It can be
checked easily that its third derivative is negative in $[1/2,1]$. Then the first
derivative of $k\log t_{1/k}\left(\mu\right)$ is concave.
$\log G_{1_{k},\frac{1}{k}}\left(\mu\right)$ have also the same properties.

The value
of $\left(-k\log t_{1/k}\left(\mu\right)\right)^{\prime}$ at the point
$\mu=1/2$ is $\log\left(2k-3\right)$. The line joining the points
$\left(1-1/k,0\right)$ to $\left(1/2,\log\left(2k-3\right)\right)$
bounds from below this first derivative. So $\left(k\log t_{1/k}\left(\mu\right)\right)^{\prime}\geq\frac{\log\left(2k-3\right)}{-\frac{1}{2}+\frac{1}{k}}\left(\mu-1+1/k\right)$.

Similarly $\left(\log k\log G_{1_{k},\frac{1}{k}}\left(\mu\right)\right)^{\prime}\geq\frac{\left(2-k\right)\log k}{\left(k-\frac{3}{2}\right)\left(-\frac{1}{2}+\frac{1}{k}\right)}\left(\mu-1+1/k\right)$. 

Summing up $\left(k\,\log\Gamma_{1_{k},\frac{1}{k},\frac{\log k}{k}}\left(\mu\right)\right)^{\prime}\geq\left(\frac{\left(2-k\right)\log k}{\left(k-\frac{3}{2}\right)\left(-\frac{1}{2}+\frac{1}{k}\right)}+\frac{\log\left(2k-3\right)}{-\frac{1}{2}+\frac{1}{k}}\right)\left(\mu-1+1/k\right)\geq0$
for $\mu\in[1/2,1-1/k[$. As a consequence: $\log\Gamma_{1_{k},\frac{1}{k},\frac{\log k}{k}}\left(\mu\right)$ is increasing in this interval and then  $\log\Gamma_{1_{k},\frac{1}{k},\frac{\log k}{k}}\left(\mu\right)<\Gamma_{1_{k},\frac{1}{k},\frac{\log k}{k}}\left(1-1/k\right)=2\,\log\gamma_{1_{k},\frac{\log k}{k}}\left(1/k\right)$
for $\mu\in[1/2,1-1/k[$. 

 We prove in the following that in this interval,
$k\log\Gamma_{1_{k},\frac{1}{k},\frac{\log k}{k}}\left(\mu\right)$
is strictly decreasing. As already seen the first derivative of $-k\log t_{1/k}\left(\mu\right)$
is concave and can be bounded from above by its tangent in $1-1/k$.
Then $\left(k\log t_{1/k}\left(\mu\right)\right)^{\prime}\leq-\frac{k^{3}}{(k-1)^{2}}\left(\mu-1+1/k\right)$.
$\log k\,\log G_{1_{k},\frac{1}{k}}\left(\mu\right)$ have also the
same properties $\left(\log k\,\log G_{1_{k},\frac{1}{k}}\left(\mu\right)\right)^{\prime}\leq\frac{k^{3}\log k}{\left(k-1\right)^{3}}\left(\mu-1+1/k\right)$. 

Summing up $ $$\left(k\,\log\Gamma_{1_{k},\frac{1}{k},\frac{\log k}{k}}\left(\mu\right)\right)^{\prime}\leq\frac{k^{3}}{\left(k-1\right)^{2}}\left(\frac{\log k}{k-1}-1\right)\left(\mu-1+1/k\right)<0$
for $k\geq3$. As a consequence: $\log\Gamma_{1_{k},\frac{1}{k},\frac{\log k}{k}}\left(\mu\right)<\Gamma_{1_{k},\frac{1}{k},\frac{\log k}{k}}\left(1-1/k\right)=2\,\log\left(\gamma_{1_{k},\frac{\log k}{k}}\left(1/k\right)\right)$
for $\mu\in]1-1/k,1]$. 
\end{proof}

\end{appendix}
\end{document}